\documentclass{article}
\usepackage{graphicx,psfrag,subfigure}
\usepackage{latexsym}
\usepackage{amsfonts, amssymb}
\usepackage{amsmath,amsthm,amssymb}
\usepackage{array, longtable}
\newtheorem{Theorem}{Theorem}[section]
\newtheorem{Lemma}[Theorem]{Lemma}
\newtheorem{Remark}[Theorem]{Remark}
\newtheorem{Definition}[Theorem]{Definition}
\newtheorem{Corollary}[Theorem]{Corollary}

\newtheorem{Example}[Theorem]{Example}

\numberwithin{equation}{section}
\numberwithin{table}{section}

\begin{document}

\title{Galois Self-Dual  Constacyclic Codes}

\insert\footins{\footnotesize
{\it Email addresses}:  yfan@mail.ccnu.edu.cn (Y. Fan)}

\author{Yun Fan\quad and\quad Liang Zhang\\
\small Dept of Mathematics,
\small  Central China Normal University, Wuhan 430079, China}

\date{}

\maketitle

\begin{abstract}
Generalizing Euclidean inner product and Hermitian inner product,
we introduce Galois inner products, and study the Galois self-dual
constacyclic codes in a very general setting by a uniform method.
The conditions for existence of Galois self-dual and
isometrically Galois self-dual constacyclic codes are obtained.
As consequences, the results on self-dual, iso-dual and Hermitian self-dual
constacyclic codes are derived.


\end{abstract}

\section{Introduction}
Constacyclic codes over finite fields are a generalization
of cyclic codes over finite fields, and inherit most of the
advantages of cyclic codes. They can be theoretically studied with polynomials,
and can be performed by feed-back shift registers in practice.
There have been many references about the constacyclic codes.
We are concerned with the research related to the duality and self-duality
of constacyclic codes.

Let ${\Bbb F}_q$ be a finite field with $q=p^e$ elements, where $p$ is a prime.
Let $\lambda$ be a non-zero element of ${\Bbb F}$, and $n$ be a positive integer.
As usual, ${\Bbb F}_q[X]$ denotes the polynomial ring.
Each element $\sum_{i=0}^{n-1}a_iX^i$ of
the quotient ring ${\Bbb F}_q[X]/\langle X^n-\lambda\rangle$
is identified with a word $(a_0,a_1,\cdots,a_{n-1})\in{\Bbb F}_q^n$.
Any ideal $C$ of ${\Bbb F}_q[X]/\langle X^n-\lambda\rangle$
is called a {\em $\lambda$-constacyclic code} of length~$n$ over ${\Bbb F}_q$.
The $1$-constacyclic codes are just the cyclic codes.
The $(-1)$-constacyclic codes are also called {\em negacyclic codes}.
If the greatest common divisor $\gcd(n,p)\!=\!1$, then
$X^n-\lambda$ has no repeated (multiple) roots and
${\Bbb F}_q[X]/\langle X^n-\lambda\rangle$ is a semisimple ring.

At the semisimple case, there are no self-dual cyclic codes.
Leon, Masley and Pless \cite{LMP} started the research on {\em duadic}
and {extended} self-dual cyclic codes. Since then,
duadic cyclic codes and various generalizations were investigated extensively, e.g.,
Pless \cite{P87}, Smid \cite{S}, Rushanan \cite{R}, Ding and Pless \cite{DP}
studied the duadic and extended self-dual cyclic codes.
Brualdi and Pless \cite{BP},
Ward and Zhu \cite{WZ}, Ling and Xing \cite{LX},
Sharma, Bakshi, Dumir and Raka \cite{SBDR}
studied the polyadic cyclic or abelian codes. 
Williams \cite{W}, Matinnes-P\'erez and Williams \cite{MW},
Fan and Zhang \cite{FZg}, Jitman, Ling and Sol\'e \cite{JLS}
studied self-dual or Hermitian self-dual group codes.

Dinh and Lopez-Permouth \cite{DL}, Dinh \cite{D}
studied constacyclic codes;
in particular, they showed that in the semisimple case
self-duality happens for and only for negacyclic codes.
Lim \cite{L} studied polyadic consta-abelian codes.
Blackford \cite{Bl08} gave conditions for the existence of
the so-called {\em Type-I duadic} negacyclic codes.
Chen, Fan, Lin and Liu \cite{CFLL} introduced a class of isometries
to classify constacyclic codes.
Blackford \cite{Bl13} introduced {\em isometrically self-dual}
(``{\em iso-dual}'' briefly) constacyclic codes, which
are proved to be just the Type-I duadic constacyclic codes.
Chen, Dinh, Fan and Ling \cite{CDFL} exhibited necessary and sufficient
conditions for the existences of polyadic constacyclic codes.
Fan and Zhang~\cite{FZl} classified the so-called {\em Type-II duadic}
constacyclic codes which are in fact isometrically
maximal self-orthogonal constacyclic codes.
Note that most of the studies mentioned above considered the semisimple case;
and, even in this case, there are less results on Hermitian self-dual constacyclic codes.

In this paper we study the duality and self-duality of constacyclic codes
in a more general setting and by a uniform method.
First, we consider any constacyclic codes, without the assumption ``$\gcd(n,p)=1$''.
Second, we define more general {\em isometries} between
constacyclic codes which may be not semisimple.
Third, we introduce a kind of inner products, called {\em Galois inner products},
as follows: for each integer $h$ with $0\le h<e$ (recall that $q=p^e$), define:
\begin{equation}\label{galois inner}
 \langle {\bf a},{\bf b}\rangle_h\!=\!\sum_{i=0}^{n-1}a_ib_i^{p^h},~~
 \forall~{\bf a}\!=\!(a_0,a_1,\cdots,a_{n-1}),
  {\bf b}\!=\!(b_0,b_1,\cdots,b_{n-1}) \in{\Bbb F}_q^n.
\end{equation}
It is just the usual Euclidean inner product if $h=0$.
And, it is the Hermitian inner product if $e$ is even and $h=\frac{e}{2}$.
Then the {\em Galois dual codes} of constacyclic codes,
the {\em Galois self-dual} (and {\em isometrically Galois self-dual})
constacyclic codes are naturally defined,
which are investigated in this paper.

Since ``$p{\,|\,}n$'' is allowed, constacyclic codes are no longer characterized
by sets of zeros. In Section 2, we introduce so-called {\em $q$-coset functions}
to characterize constacyclic codes.

We'll study the {\em isometrically Galois self-dual}
constacyclic codes in our general setting. So, in Section 3,
we define the isometries between constacyclic codes in the general setting
and explore their properties. The main result is Theorem \ref{M_s phi} below.

In Section 4, with the isometries introduced in Section 3
we characterize the Galois dual codes of constacyclic codes by
$q$-coset functions. The main result is Theorem \ref{C_phi bot} below.
The results on dual and Hermitian dual codes are listed in
Corollary \ref{Euc-Herm} below.

In Section 5, a necessary and sufficient condition for the
existence of isometrically Galois self-dual constacyclic codes
is obtained, see Theorem \ref{iso h-self dual} below,
which covers of course the isometrically self-dual case 
and the isometrically Hermitian self-dual case.

In Section 6, we study Galois self-dual constacyclic codes, and show 
a necessary and sufficient condition for their existence,
see Theorem \ref{h-self dual} below.
The existence results on self-dual constacyclic codes and
Hermitian self-dual constacyclic codes are drawn as consequences,
see Corollary \ref{self-dual} and Corollary \ref{Hermitian self-dual} below.

Finally, some examples are illustrated in Section 7.

\section{Constacyclic codes and $q$-coset functions}

In this paper we always take the following notations:
\begin{itemize}
\item
${\Bbb F}_q$ denotes the finite field
with cardinality $|{\Bbb F}_q|=q=p^e$,
where $p$ is a prime and $e$ is a positive integer,
and ${\Bbb F}_q^*$ denotes the multiplicative group
consisting of units of ${\Bbb F}_q$.
So ${\Bbb F}_q^*$ is a cyclic group of order $q-1$.
\item
$n$ is any positive integer, $\nu_p(n)$ denotes the $p$-adic valuation of $n$;
hence $n=p^{\nu_p(n)} n'$ with $n'$ being coprime to $p$.
\item
$h\in[0,e]$, where $[0,e]=\{0,1,\cdots,e\}$ is an integer interval,
and $\langle{\bf a},{\bf b}\rangle_h=\sum_{i=0}^{n-1}a_ib_i^{p^h}$
as in Eqn \eqref{galois inner}.
\item
$\lambda\in {\Bbb F}_q^*$ with
${\rm ord}_{{\Bbb F}_q^*}(\lambda)=r$,
where ${\rm ord}_{{\Bbb F}_q^*}(\lambda)$ denotes the order of 
$\lambda$ in the group ${\Bbb F}_q^*$, hence $r\,\big|\,(q-1)$.
\item
$R_{n,\lambda}={\Bbb F}_q[X]/\langle X^n-\lambda\rangle$
is the quotient ring of the polynomial ring
${\Bbb F}_q[X]$ over ${\Bbb F}_q$ with respect to the ideal
$\langle X^n-\lambda\rangle$ generated by $ X^n-\lambda$.
By $C\le R_{n,\lambda}$ we mean that $C$ is an ideal of $R_{n,\lambda}$, i.e.,
$C$ is a $\lambda$-constacyclic code of length $n$ over ${\Bbb F}_q$.
\end{itemize}

\begin{Remark}\label{galois}\rm
By ${\Bbb Z}_e$ we denote the residue ring of
the integer ring ${\Bbb Z}$ modulo $e$.
Then the additive group of ${\Bbb Z}_e$ is isomorphic to
the Galois group of ${\Bbb F}_q$ over ${\Bbb F}_p$ by
mapping $h\in {\Bbb Z}_e$ to the Galois automorphism
$\gamma_{p^h}$ of ${\Bbb F}_q$, where
$\gamma_{p^h}(a)=a^{p^h}$ for all $a\in{\Bbb F}_q$.
So, we call $\langle{\bf a},{\bf b}\rangle_h$ a {\em Galois inner product}
on ${\Bbb F}_q^n$.
\end{Remark}

Any element of the quotient ring $ R_{n,\lambda}$ has a unique representative
of degree at most $n-1$: $a(X)=a_0+a_1X+\cdots+a_{n-1}X^{n-1}$.
We always associate any word
$a=(a_0,a_1,\cdots,a_{n-1})\in{\Bbb F}_q^n $ with
$a(X)=a_0+a_1X+\cdots+a_{n-1}X^{n-1}$ of the ring $R_{n,\lambda}$,
and {\it vice versa}. Hence the Hamming weight ${\rm w}(a(X))$
for $a(X)\in R_{n,\lambda}$ and the minimal weight ${\rm w}(C)$
for $C\le R_{n,\lambda}$ are defined as usual.

By Remark \ref{galois},
there is a unique $\lambda'\in{\Bbb F}_q$ such that
$\lambda=\gamma_{p^{\nu_p(n)}}(\lambda')
={\lambda'}^{p^{\nu_p(n)}}$, hence
${\rm ord}_{{\Bbb F}_q^*}(\lambda')
 ={\rm ord}_{{\Bbb F}_q^*}(\lambda)=r$.
Note that $\gcd(q,n'r)=1$,
where $\gcd$(-,-) denotes the greatest common divisor.
In the following we always assume that:
\begin{itemize}
\item
$\theta$  is a primitive  $n'r$-th  root of unity in
${\Bbb F}_{q^d} $ (with $d={\rm ord}_{{\Bbb Z}_{n'r}^*}(q)$)
such that $\theta^{n'}=\lambda'$ (equivalently, $\theta^{n}=\lambda$),
where ${\Bbb Z}_{n'r}^*$ denotes the
multiplicative group consisting of units of the residue ring ${\Bbb Z}_{n'r}$.
\item
$1+r{\Bbb Z}_{n'r}$ is
the subset of $ {\Bbb Z}_{n'r}$  as follows:
$$
 1+r{\Bbb Z}_{n'r}
 =\big\{1+rk\hskip-6pt\pmod{n'r}\,\big|\,k\in{\Bbb Z}_{n'r}\big\}
  =\{1,1+r,\cdots,1+r(n'-1)\}.
$$
\end{itemize}
It is easy to check that $\theta^i$ for all $i\in (1+r{\Bbb Z}_{n'r})$
are all roots of $X^{n'}-\lambda'$. In ${\Bbb F}_{q^d}[X]$
we have the following decomposition:
\begin{equation}\label{in F_{q^d}}
 X^n-\lambda=(X^{n'}-\lambda')^{p^{\nu_p(n)}}
 =\prod\limits_{i\in (1+r{\Bbb Z}_{n'r})}(X-\theta^i)^ {p^{\nu_p(n)}}.
\end{equation}

Let $s$ be an integer with $\gcd(s,n'r)=1$. Then $s$ induces a bijection 
\begin{equation}\label{mu_s}
\mu_s:~~ 1+r{\Bbb Z}_{n'r}~\to~ s+r{\Bbb Z}_{n'r},~~
 k\mapsto sk \pmod{n'r},
\end{equation}
where
\begin{equation}\label{s+}
 s+r{\Bbb Z}_{n'r}
 =\big\{s+rk\hskip-5pt\pmod{n'r}\,\big|\,k\in{\Bbb Z}_{n'r}\big\},
\end{equation}
and $\theta^i$ for all $i\in (s+r{\Bbb Z}_{n'r})$
are all roots (with multiplicity $p^{\nu_p(n)}$) of $X^{n}-\lambda^s$.

It is easy to see that $s+r{\Bbb Z}_{n'r}=1+r{\Bbb Z}_{n'r}$
if and only if $s\equiv 1\!\pmod{r}$.
Assume that $s\equiv 1\!\pmod{r}$.
Then $\mu_s$ is a permutation of $1+r{\Bbb Z}_{n'r}$.
Any orbit of the permutation is called an 
{\em $s$-orbit} on $1+r{\Bbb Z}_{n'r}$.
In fact, for any integer $t$ coprime to $n'r$, 
the $\mu_s$ (with $s\equiv1\!\pmod r$)
is a permutation of the set $t+r{\Bbb Z}_{n'r}$, 
which is then partitioned into $s$-orbits.

\begin{Remark}\label{q-coset}\rm
(i)~ Since $r\,\big|\,(q-1)$, we have $\gcd(q,n'r)=1$
and $q\equiv 1\!\pmod r$.
The $q$-orbits on $1+r{\Bbb Z}_{n'r}$ are also named
{\em $q$-cyclotomic cosets}, or {\em $q$-cosets} in short. 
The quotient set consisting of $q$-cosets on $1+r{\Bbb Z}_{n'r}$
is denoted by $(1+r{\Bbb Z}_{n'r})/\mu_q$.

(ii)~ For an integer $s$ with $\gcd(s,n'r)=1$ and $s\equiv 1\!\pmod r$,
the permutation $\mu_s$ on $1+r{\Bbb Z}_{n'r}$
induces a permutation, denote by $\mu_s$ again,
of the quotient set $(1+r{\Bbb Z}_{n'r})/\mu_q$,
and partitions the quotient set into $s$-orbits; cf \cite[Lemma 8]{CDFL}.
That is, for any $Q\in (1+r{\Bbb Z}_{n'r})/\mu_q$,
$sQ=\{sk\,|\,k\in Q\}$ is still a $q$-coset, and
there is a positive integer $\ell$ such that
$s^iQ\ne s^jQ$ if $0\le i\ne j<\ell$, but $s^\ell Q=Q$;
then $Q$, $sQ$, $\cdots$, $s^{\ell-1}Q$ form
an $s$-orbit of {length} $\ell$ on the quotient set.

Example \ref{ex q=25} in Section 7 is a non-trivial example 
for the above notations.
\end{Remark}

Then we further define the polynomials $f_Q(X)$ for $q$-cosets $Q$'s 
and get a decomposition 
as follows: 
\begin{equation}\label{f_Q}
  X^n-\lambda=\prod\limits_{Q\in (1+r{\Bbb Z}_{n'r})/{\mu_q}}
  f_Q(X)^ {p^{\nu_p(n)}} ~~~
  \mbox{where}~ f_Q(X)=\prod\limits_{i\in Q}(X-\theta^i),
\end{equation}
and $f_Q(X)$ for all $Q\in (1+r{\Bbb Z}_{n'r})/{\mu_q}$
are irreducible ${\Bbb F}_{q}$-polynomials.

\begin{Definition}\label{def q-coset function}\rm
Let $[0,p^{\nu_p(n)}]$ be the integer interval $\{0,1,\cdots,p^{\nu_p(n)}\}$.
\begin{itemize}
\item[(i)]
A map $\varphi:1+r{\Bbb Z}_{n'r}\rightarrow[0,p^{\nu_p(n)}]$ is called a
{\em $q$-coset function} if for any $k\in  1+r{\Bbb Z}_{n'r}$
and any integer $i$ we have $\varphi(q^i k)=\varphi(k)$.
Thus, any $q$-coset function
$\varphi:1+r{\Bbb Z}_{n'r}\rightarrow[0,p^{\nu_p(n)}]$
is identified with a function (denoted by $\varphi$ again) on the quotient set:
$$  \varphi:(1+r{\Bbb Z}_{n'r})/{\mu_q}\to [0,p^{\nu_p(n)}],~
 Q\mapsto \varphi(Q), ~
 \mbox{where $\varphi(Q)\!=\!\varphi(k)$ for $k\in Q$};
$$
then a polynomial $f_\varphi(X)\in{\Bbb F}_q[X]$ can be defined as follows:
$$ f_\varphi(X)=\prod\limits_{Q\in (1+r{\Bbb Z}_{n'r})/\mu_q}f_Q(X)^ {\varphi(Q)}.  $$
\item[(ii)]
For any $q$-coset function $\varphi:1+r{\Bbb Z}_{n'r}\rightarrow[0,p^{\nu_p(n)}]$,
we define a function
$\bar{\varphi}:1+r{\Bbb Z}_{n'r}\rightarrow[0,p^{\nu_p(n)}]$
by $\bar{\varphi}(k)=p^{\nu_p(n)}-{\varphi}(k)$.
The function $\bar{\varphi}$ is also a $q$-coset function, which
we call by the {\em complement of $\varphi$}.
\item[(iii)]
For any integer $s$ coprime to $n'r$ and any $q$-coset function
$\varphi:1+r{\Bbb Z}_{n'r}\rightarrow[0,p^{\nu_p(n)}]$,
define a function
$s\varphi:s+r{\Bbb Z}_{n'r}\rightarrow[0,p^{\nu_p(n)}]$ by
$$ (s\varphi)(k)=\varphi(s^{-1}k),~~~~ \forall~ k\in s+r{\Bbb Z}_{n'r}.$$
where $s^{-1}$ is an integer such that
$s^{-1}s\equiv 1\!\pmod{n'r}$.
It is easy to check that $s\varphi$ is still a $q$-coset function.
\item[(iv)]
Let $\varphi,\varphi':1+r{\Bbb Z}_{n'r}\rightarrow[0,p^{\nu_p(n)}]$
be $q$-coset functions. Define
$$
 (\varphi\cap\varphi')(k)=
  \min\{\varphi(k),\varphi'(k)\}, ~~~~  \forall~ k\in 1+r{\Bbb Z}_{n'r}.
$$
The function $ \varphi\cap\varphi'$ is clearly a $q$-coset function too.
Further, if $ \varphi\cap\varphi'=\varphi$, then we write $ \varphi\leq\varphi'$.
\end{itemize}
\end{Definition}

By Eqn \eqref{f_Q} and the definition of $\varphi\cap\bar\varphi$, we have
\begin{equation}
f_\varphi(X)f_{\bar{\varphi}}(X)=X^n-\lambda, ~~~~
\gcd\big(f_\varphi(X),f_{\bar{\varphi}}(X)\big)
 =f_{\varphi\cap\bar\varphi}(X).
\end{equation}

It is a routine to verify the following two lemmas. 

\begin{Lemma}\label{generator}
For any $\lambda$-constacyclic code $C\leq R_{n,\lambda}$
there is exactly one  $q$-coset function
$\varphi:(1+r{\Bbb Z}_{n'r})/\mu_q\rightarrow[0,p^{\nu_p(n)}]$
satisfying the following:
\begin{itemize}
\item[\bf(i)]
$c(X)\in C$ if and only if $c(X)f_{\varphi}(X)\equiv 0\!\pmod{X^n-\lambda}$.
\item[\bf(ii)]
$c(X)\in C$ if and only if $f_{\bar{\varphi}}(X)\,|\,c(X)$.
\end{itemize}
\end{Lemma}

\begin{Definition}\label{check polynomial}\rm
As usual, $f_{\varphi}(X)$ in Lemma \ref{generator}
is called a {\em check polynomial} of the $\lambda$-constacyclic code $C$,
and $f_{\bar{\varphi}}(X)$ is called a {\em generator polynomial} of $C$.
Because of the uniqueness of the $q$-coset function $\varphi$,
we denote the $\lambda$-constacyclic code $C$ by $C_\varphi$,
and call it the $\lambda$-constacyclic code with check polynomial $f_\varphi(X)$.
\end{Definition}

For any $C\le R_{n,\lambda}$, we set
\begin{equation}\label{Ann}
{\rm Ann}(C)=\big\{ a(X)\in R_{n,\lambda}\,\big|\,
   a(X)c(X)\!\equiv\! 0\hskip-6pt \pmod{X^n-\lambda},~ \forall~ c(X)\in C \big\}.
\end{equation}
which is an ideal of $R_{n,\lambda}$, i.e., is a $\lambda$-constacyclic code too.

\begin{Lemma}\label{Ann=bar}
 Let $C_\varphi\leq R_{n,\lambda}$,
 where $\varphi:(1+r{\Bbb Z}_{n'r})/\mu_q\rightarrow[0,p^{\nu_p(n)}]$
 is a $q$-coset function.
 Then ${\rm Ann}(C_\varphi)=C_{\bar{\varphi}}$.
\end{Lemma}


 \section{Ring isometries}

For any non-zero integer $s$,
by $\nu_p(s)$  we denote the $p$-adic valuation of $s$, hence
$s=p^{\nu_p(s)}s'$ with $p\!\nmid\! s'$.
If $\gcd(s,n'r)=1$ then $\gcd(s',nr)=1$ obviously.

\begin{Theorem}\label{ring iso}
Assume that $\gcd(s,n'r)=1$, $s=p^{\nu_p(s)}s'$ and
$s'^{-1}$ is an integer such that $s'^{-1}s'\equiv 1\!\pmod{nr}$. Then the map
${\cal M}_s:  R_{n,\lambda}\to R_{n,\lambda^s}$ defined by
\begin{equation}\label{M_s}
{\cal M}_s\Big(\sum_{i=0}^{n-1}a_iX^i\Big)=
 \sum_{i=0}^{n-1}a_i^{p^{\nu_p(s)}}X^{is'^{-1}}\hskip-8pt
 \pmod{X^n-\lambda^s},~~~
  \forall~ \sum_{i=0}^{n-1}a_iX^i\in R_{n,\lambda},
\end{equation}
is well-defined (independent of the choice of $s'^{-1}$) and the following hold:
\begin{enumerate}
  \item[\bf(i)] ${\cal M}_s$ is a ring isomorphism.
  \item[\bf(ii)]${\rm w}\big( {\cal M}_s(a(X))\big)={\rm w}\big((a(X))\big)$
  for all $a(X)\in R_{n,\lambda}$.
\end{enumerate}
\end{Theorem}

\begin{proof}
Mapping $a$ to $a^{p^{\nu_p(s)}}$ is an automorphism of ${\Bbb F}_q$, 
see Remark~\ref{galois}. It is obvious that the following map
$$
 {\Bbb F}_q [X]\rightarrow{\Bbb F}_q [X],~~~
 \sum\limits_{i=0}^{k}a_iX^i\mapsto\sum_{i=0}^{k}a_i^{p^{\nu_p(s)}}X^{is'^{-1}}
$$
is a ring homomorphism; hence it induces a ring homomorphism:
 \begin{equation*}
\begin{array}{rcl}
 \widehat{\cal M}_s:{\Bbb F}_q [X] &\rightarrow&
  {\Bbb F}_q [X]/\langle X^n-\lambda^s\rangle,\\
 \sum\limits_{i=0}^{k}a_iX^i & \rightarrow &
  \sum\limits_{i=0}^{k}a_i^{p^{\nu_p(s)}}X^{is'^{-1}} \pmod{X^n-\lambda^s}.
\end{array}
\end{equation*}
In the ring $R_{n,\lambda^s}= {\Bbb F}_q [X]/\langle X^n-\lambda^s\rangle$
 we have the following computation:
$$\widehat{\cal M}_s(X^n-\lambda)=X^{ns'^{-1}}\!-\lambda^{p^{\nu_p(s)}}
\!\equiv \,(\lambda^{p^{\nu_p(s)}s'})^{s'^{-1}}\!-\lambda^{p^{\nu_p(s)}}
\!= 0\! \pmod{X^n-\lambda^s}.$$
Thus the ring homomorphism $\widehat{\cal M}_s$
induces a well-defined ring homomorphism as follows:
\begin{equation*}
\begin{array}{rcl}
{\cal M}_s:{\Bbb F}_q [X]/\langle X^n-\lambda\rangle &\rightarrow&
  {\Bbb F}_q [X]/\langle X^n-\lambda^s\rangle,\\
\sum\limits_{i=0}^{n-1}a_iX^i & \rightarrow & \sum\limits_{i=0}^{n-1}a_i^{p^{\nu_p(s)}}X^{is'^{-1}} \pmod{X^n-\lambda^s}.
\end{array}
\end{equation*}
Because $s'^{-1}$ is unique modulo $nr$ and $\lambda^{nr}=1$,
${\cal M}_s$ is independent of the choice of the integer $s'^{-1}$
such that $s'^{-1}s'\equiv 1\!\pmod{nr}$.
Since $\gcd(s'^{-1},n)=1$, for any $j$ we can find an $i$
such that $is'^{-1}\equiv j\!\pmod n$, i.e.,
$is'^{-1}=nt+j$ for an integer $t$.
Further, there is an $a\in {\Bbb F}_q$  such that
$a^{p^{\nu_p(s)}}=\lambda^{-st}$.
Then in the ring ${\Bbb F}_q [X]/\langle X^n-\lambda^s\rangle$ we have
$$ {\cal M}_s(aX^i)=a^{p^{\nu_p(s)}}X^{is'^{-1}}
 \equiv a^{p^{\nu_p(s)}}\lambda^{st}X^j=X^j  \pmod{X^n-\lambda^s}.
$$
 Thus  the ring homomorphism  ${\cal M}_s$ is surjective.
 Further, the cardinalities of ${\Bbb F}_q [X]/\langle X^n-\lambda\rangle$
 and ${\Bbb F}_q [X]/\langle X^n-\lambda^s\rangle$
 are equal to each other. So ${\cal M}_s$ is a ring isomorphism, i.e.,
 (i) holds. Finally, (ii) holds obviously.
\end{proof}

\begin{Remark}\rm
We call the ${\cal M}_s$ defined in Theorem \ref{ring iso}
a {\em ring isometry} from $R_{n,\lambda}$ to $R_{n,\lambda^s}$.
Note that the isometries between constacyclic codes appeared in literature,
e.g., in \cite{CFLL, Bl13, CDFL, FZl}, are defined only for the semisimple case
(i.e., $\gcd(n,q)=1$) and are algebra isomorphism.
The ring isometries ${\cal M}_s$ in Theorem~\ref{ring iso}
are defined for the general case (where ``$p\,|\,n$'' is allowed),
and are semi-linear algebra isomorphisms in general, i.e.,
they are isomorphisms of rings and 
semi-linear isomorphisms of vector spaces.
Precisely, ${\cal M}_s$ is a 
$\gamma_{p^{\nu_p(s)}}$-linear isomorphism, 
where $\gamma_{p^{\nu_p(s)}}$ is the 
Galois automorphism defined in Remark \ref{galois}.
For any constacyclic code $C\le R_{n,\lambda}$,
by the semi-linearity of ${\cal M}_s$, we still have
$\dim_{{\Bbb F}_q}\big({\cal M}_s(C)\big)=\dim_{{\Bbb F}_q}(C)$.
\end{Remark}

\begin{Lemma}\label{s_1=s_2?}
Let $s_1$ and $s_2$ be integers coprime to $n'r$,
let $s_1=p^{\nu_p(s_1)}s_1'$ and $s_2=p^{\nu_p(s_2)}s_2'$.
Then the following two are equivalent:
\begin{itemize}
\item[\bf(i)]
${\cal M}_{s_1}={\cal M}_{s_2}$.
\item[\bf(ii)]
$s_1'\equiv s_2'\!\pmod{nr}$~ and~
$\nu_p(s_1)\equiv \nu_p(s_2)\!\pmod{e}$.
\end{itemize}
\end{Lemma}

\begin{proof}
Let $s_1'^{-1}$, $s_2'^{-1}$ be integers with
$s_1'^{-1}s_1'\equiv 1\!\pmod{nr}$, $s_2'^{-1}s_2'\equiv 1\!\pmod{nr}$.
Suppose that ${\cal M}_{s_1}={\cal M}_{s_2}$. Then
$\lambda^{s_1}=\lambda^{s_2}$, hence $s_1\equiv s_2\pmod{r}$.
And, ${\cal M}_{s_1}(X)={\cal M}_{s_2}(X)$, i.e.,
$X^{{s'_1}^{-1}}\equiv X^{{s'_2}^{-1}}\!\pmod{X^n-\lambda^{s_1}}$.
Let ${s'_i}^{-1}=t_i n+k_i$  with $0\le k_i<n$ for $i=1,2$. Then
$X^{{s'_i}^{-1}}\equiv \lambda^{s_1 t_i}X^{k_i}\!
 \pmod{X^n-\lambda^{s_1}}$.
We get that $\lambda^{s_1 t_1}X^{k_1}=\lambda^{s_1 t_2}X^{k_2}$.
Then $t_1\equiv t_2\!\pmod r$ and $k_1=k_2$, which imply that
${s'_1}^{-1}\equiv {s'_2}^{-1}\!\pmod{nr}$, equivalently,
${s'_1}\equiv {s'_2}\!\pmod{nr}$.
Further,
${\cal M}_{s_1}(a)={\cal M}_{s_2}(a)$ for any $a\in{\Bbb F}_q$.
Then $a^{p^{\nu_p(s_1)}}=a^{p^{\nu_p(s_2)}}$ for any $a\in{\Bbb F}_q$;
so $\nu_p(s_1)\equiv\nu_p(s_2)\pmod{e}$, cf. Remark \ref{galois}.
The necessity is proved. The sufficiency is obvious.
\end{proof}

It is easy to see that 
\begin{equation}\label{s1 s2}
 {\cal M}_{s_1}{\cal M}_{s_2}={\cal M}_{s_1s_2},~~~~
 \mbox{for $s_1,s_2$ coprime to $n'r$.}
\end{equation}
Lemma~\ref{s_1=s_2?} implies that
the set of ring isometries
 $\big\{{\cal M}_s\,\big|\, \mbox{$s$ is coprime to $n'r$}\big\}$
 form a group which is isomorphic to ${\Bbb Z}_e\times{\Bbb Z}_{n'r}^*$.

\begin{Remark}\label{s+rZ}\rm We know that
$\lambda^s=\lambda$ if and only if  $s\equiv 1\!\pmod{r}$.
So $R_{n,\lambda^s}=R_{n,\lambda}$ 
if and only if $s\in 1+r{\Bbb Z}_{n'r}$.
At that case, for any $Q\in(1+r{\Bbb Z}_{n'r})/\mu_q$,
$sQ=\{sk\,|\,k\in Q\}$ is still a $q$-coset on $1+r{\Bbb Z}_{n'r}$.
However, if $s{\not\equiv 1}\!\pmod{r}$,
then $\theta^i$ for $i\in s+r{\Bbb Z}_{n'r}$ are
all roots of $X^{n'}-\lambda'^s$, cf. Eqn \eqref{s+}.
And, a $\mu_q$-action on $s+r{\Bbb Z}_{n'r}$ 
is defined the same as in  Eqn~\eqref{mu_s}
so that, for any $Q\in(1+r{\Bbb Z}_{n'r})/\mu_q$,
the $sQ$ is a {$q$-coset} on $s+r{\Bbb Z}_{n'r}$; cf. Ramark \ref{q-coset}.
Thus $Q\mapsto sQ$ is a bijective map 
from $(1+r{\Bbb Z}_{n'r})/\mu_q$ onto $(s+r{\Bbb Z}_{n'r})/\mu_q$; 
the converse map sends any $q$-coset $Q'\in (s+r{\Bbb Z}_{n'r})/\mu_q$ 
to $s^{-1}Q'=\{s^{-1}k'\,|\,k'\in Q'\}$, 
where $s^{-1}s\equiv 1\!\pmod{n'r}$
as in Definition \ref{def q-coset function} (iii).
\end{Remark}

\begin{Lemma}\label{M_s f_Q}
Let $s$ be an integer coprime to $n'r$. 
Then for any $q$-coset $Q\in(1+r{\Bbb Z}_{n'r})/\mu_q$ there is a unit
$u(X)\in R_{n,\lambda^s}$ such that
$$  {\cal M}_s\big(f_Q(X)\big)=u(X) f_{sQ}(X).$$
\end{Lemma}

\begin{proof} Let $s=p^{\nu_p(s)}s'$ and $s'^{-1}$ be an integer such that
$s'^{-1}s'\equiv 1\!\pmod{nr}$.
Note that $f_Q(X)=\prod_{i\in Q}(X-\theta^i)$, see Eqn \eqref{f_Q},
and ${\cal M}_s$ is a ring isomorphism, see Theorem \ref{ring iso}. So
\begin{eqnarray*}
{\cal M}_s\big(f_Q(X)\big)
 &=&\prod\limits_{i\in Q}(X^{{s'}^{-1}}-\theta^{ip^{\nu_p(s)}})
    =\prod\limits_{i\in Q}(X^{{s'}^{-1}}-(\theta^{is})^{{s'}^{-1}})\\
  &=& \Big(\prod\limits_{i\in Q}(X-\theta^{is})\Big)\Big( \prod\limits_{i\in Q}
   \frac{X^{{s'}^{-1}}-(\theta^{is})^{{s'}^{-1}}}{X-\theta^{is}}\Big) \\
  &=& \Big(\prod\limits_{j\in sQ}(X-\theta^{j})\Big)\Big(\prod\limits_{i\in Q}
   \frac{X^{{s'}^{-1}}-(\theta^{is})^{{s'}^{-1}}}{X-\theta^{is}}\Big) \\
   &=& f_{sQ}(X)\cdot u(X)\,,
\end{eqnarray*}
where $u(X)=\prod\limits_{i\in Q}
   \frac{X^{{s'}^{-1}}-(\theta^{is})^{{s'}^{-1}}}{X-\theta^{is}}$.
It is enough to show that $u(X)$ is coprime to $X^n-\lambda^s$.
Further, it is enough to show that, for any $i\in Q$,
$\frac{X^{{s'}^{-1}}-(\theta^{is})^{{s'}^{-1}}}{X-\theta^{is}}$
is coprime to $X^n-\lambda^s$.
Note that $\theta^{js}$ for $j\in 1+r{\Bbb Z}_{n'r}$
are all roots of $X^n-\lambda^s$, cf. Eqn \eqref{s+} and Remark \ref{s+rZ}.
If $j\not\equiv i\! \pmod{n'r}$,
then $js{s'}^{-1}\not\equiv is{s'}^{-1}\! \pmod{n'r}$
because $ss'^{-1}$ is coprime to $n'r$,
hence $(\theta^{js})^{{s'}^{-1}}-(\theta^{is})^{{s'}^{-1}}\ne 0 $.
So, any root $\theta^{js}$ of $X^n-\lambda^s$ for $j\in 1+r{\Bbb Z}_{n'r}$
is not a root of the polynomial
$\frac{X^{{s'}^{-1}}-(\theta^{is})^{{s'}^{-1}}}{X-\theta^{is}}$.
This completes the proof of the lemma.
\end{proof}

\begin{Lemma}
Let $s$ be an integer coprime to $n'r$, and
$\varphi:1+r{\Bbb Z}_{n'r}\to[0,p^{\nu_p(n)}]$ be a $q$-coset function. 
Then there is a unit $u(X)\in R_{n,\lambda^s}$ such that
$$  {\cal M}_s\big(f_\varphi(X)\big)=u(X) f_{s\varphi}(X).$$
\end{Lemma}

\begin{proof}
Since $f_\varphi(X)
 =\prod_{Q\in(1+r{\Bbb Z}_{n'r})/\mu_q}f_Q(X)^{\varphi(Q)}$,
 see Definition \ref{def q-coset function} (i),
and ${\cal M}_s$ is a ring isomorphism, we have:
$$
{\cal M}_s\big(f_\varphi(X)\big)
 =\prod_{Q\in(1+r{\Bbb Z}_{n'r})/\mu_q}
    {\cal M}_s\big(f_Q(X)\big)^{\varphi(Q)}.
$$
For each $Q\in(1+r{\Bbb Z}_{n'r})/\mu_q$, by Lemma \ref{M_s f_Q}
we have a unit $u_Q(X)$ in $R_{n,\lambda^s}$ such that
 ${\cal M}_s\big(f_Q(X)\big)=u_Q(X)f_{sQ}(X)$.
Then $u(X)=\prod_{Q\in(1+r{\Bbb Z}_{n'r})/\mu_q}u_Q(X)$
is a unit of $R_{n,\lambda^s}$ and
$$
{\cal M}_s\big(f_\varphi(X)\big)
 = u(X)\prod_{Q\in(1+r{\Bbb Z}_{n'r})/\mu_q}f_{sQ}(X)^{\varphi(Q)}.
$$
As mentioned in Remark \ref{s+rZ},  $sQ$ runs over
$(s+r{\Bbb Z}_{n'r})/\mu_q$ when $Q$ runs over
$(1+r{\Bbb Z}_{n'r})/\mu_q$;
conversely, $s^{-1}Q'$ runs over
$(1+r{\Bbb Z}_{n'r})/\mu_q$ when $Q'$ runs over
$(s+r{\Bbb Z}_{n'r})/\mu_q$. Thus, we get
$$
{\cal M}_s\big(f_\varphi(X)\big)
 = u(X)\prod_{Q'\in(s+r{\Bbb Z}_{n'r})/\mu_q}f_{Q'}(X)^{\varphi(s^{-1}Q')}.
$$
But $\varphi(s^{-1}Q')=s\varphi(Q')$, 
see Definition \ref{def q-coset function} (iii). So
$$
{\cal M}_s\big(f_\varphi(X)\big)
 = u(X)\prod_{Q'\in(s+r{\Bbb Z}_{n'r})/\mu_q}f_{Q'}(X)^{(s\varphi)(Q')};
$$
that is, ${\cal M}_s\big(f_\varphi(X)\big)=u(X)f_{s\varphi}(X)$.
\end{proof}

As a consequence, we obtain the following immediately.

\begin{Theorem}\label{M_s phi}
Let $s$ be an integer coprime to $n'r$, and
$\varphi:1+r{\Bbb Z}_{n'r}\to[0,p^{\nu_p(n)}]$ be a $q$-coset function.
Then ${\cal M}_s(C_\varphi)=C_{s\varphi}$ 
which is a $\lambda^s$-constacyclic code.
\end{Theorem}

By Theorem \ref{M_s phi} and Lemma \ref{s_1=s_2?}, we get the following at once.

\begin{Corollary}\label{s_1=s_2}
Let $s_1$ and $s_2$ be integers coprime to $n'r$, let
$s_1=p^{\nu_p(s_1)}{s_1'}$ and $s_2=p^{\nu_p(s_2)}{s_2'}$.
Then the following two are equivalent to each other:
\begin{itemize}
\item[\bf(i)]
$s_1\varphi=s_2\varphi$, for any
$q$-coset function $\varphi: 1+r{\Bbb Z}_{n'r}\to[0,p^{\nu_p(n)}]$.
\item[\bf(ii)]
$s_1'\equiv s_2'\!\pmod{nr}$~ and~
 $\nu_p(s_1)\equiv \nu_p(s_2)\!\pmod{e}$.
\end{itemize}
\end{Corollary}

\section{Galois dual codes of constacyclic codes}

\begin{Definition}\label{def h-inner product}\rm
For $h\in[0,e]$, the Galois inner product
$\langle{\bf a},{\bf b}\rangle_h=\sum\limits_{i=0}^{n-1}a_ib_i^{p^h}$
is defined in Eqn \eqref{galois inner}, which
we call explicitly the {\em $p^h$-inner product}.
For any code $C\subseteq {\Bbb F}_q^n$ we have a code
$C^{\bot h}=\big\{{\bf a}\in{\Bbb F}_q^n\,\big|\,
   \langle{\bf c},{\bf a}\rangle_h=0,~\forall~{\bf c}\in C\big\},$
and call it the {\em Galois dual-code}
(more explicitly, the {\em $p^h$-dual code}) of the code $C$.
\end{Definition}

\begin{Remark}\rm
Obviously, $\langle{\bf a},{\bf b}\rangle_h$ is a non-degenerate form
on ${\Bbb F}_q^n$; it is a linear function for the first variable ${\bf a}$,
while it is a semi-linear function for the second variable~${\bf b}$;
more precisely, it is $\gamma_{p^h}$-linear for the second variable ${\bf b}$,
where $\gamma_{p^h}$ is the Galois automorphism of the
field ${\Bbb F}_q$ defined in Remark \ref{galois}.
In particular, if $C\subseteq {\Bbb F}_q^n$ is a linear code, then
$\dim_{{\Bbb F}_q}(C)+\dim_{{\Bbb F}_q}(C^{\bot h})=n$.

The $p^0$-inner product $\langle{\bf a},{\bf b}\rangle_0$ is the
Euclidean inner product and $C^{\bot0}$ is just the usual dual code of $C$.
The $p^{\frac{e}{2}}$-inner product
$\langle{\bf a},{\bf b}\rangle_{\frac{e}{2}}$ (if $e$ is even) is
the Hermitian inner product and $C^{\bot{\frac{e}{2}}}$
is just the Hermitian dual code of $C$.
\end{Remark}

In this section we characterize the Galois dual codes of constacyclic codes
by $q$-coset functions.

\begin{Lemma}
Let $C\subseteq R_{n,\lambda}$ be a code, let
${\rm Ann}(C)$ be as in Eqn \eqref{Ann}. Then
$$ C^{\bot h}={\cal M}_{-p^{e-h}}({\rm Ann}(C)).$$
In particular, $C^{\bot h}$ is a $\lambda^{-p^{e-h}}$-constacyclic code.
\end{Lemma}

\begin{proof} Assume
$a(X)=\sum\limits_{i=0}^{n-1}a_iX^i,\ b(X)=\sum\limits_{i=0}^{n-1}b_iX^i$.
In the ring $R_{n,\lambda}$ we have the following computation:
\begin{eqnarray*}
a(X)b(X)&= & \sum\limits_{k=0}^{n-1}\sum\limits_{i+j=k}a_ib_jX^k
   +\sum\limits_{k=n}^{2n-2}\sum\limits_{i+j=k}a_ib_jX^k  \\[3pt]
   &\equiv& \sum\limits_{k=0}^{n-1}\left(\sum\limits_{i+j=k}a_ib_j
    +\sum\limits_{i+j=n+k}\lambda a_ib_j\right)X^k \pmod{X^n-\lambda}.
\end{eqnarray*}
Thus
$$ a(X)b(X)\equiv 0 \pmod{X^n-\lambda}
$$
if and only if
\begin{equation}\label{in R_n,lambda}
\sum\limits_{i=0}^{k}a_ib_{k-i}+
 \lambda\sum\limits_{i=k+1}^{n-1} a_ib_{k+n-i}=0, ~~~~ k=0,1,\cdots,n-1.
\end{equation}

In $R_{n,\lambda^{-p^{e-h}}}
 ={\Bbb F}_q[X]/\langle X^n-\lambda^{-p^{e-h}}\rangle$,
consider the ideal
$\big\langle {\cal M}_{-p^{{e-h}}}(b(X))\big\rangle$
generated by the polynomial ${\cal M}_{-p^{{e-h}}}(b(X))$.
Since
$$\lambda^{p^{{e-h}}}X^n\equiv 1 \pmod{X^n-\lambda^{-p^{{e-h}}}},$$
for $0\leq k \leq n$, $X^k$ is a unit of $R_{n,\lambda^{-p^{{e-h}}}}$,
$X^{k}{\cal M}_{-p^{{e-h}}}(b(X))$ is a generator
of the ideal $\big\langle {\cal M}_{-p^{{e-h}}}(b(X))\big\rangle$ and
\begin{align*}
&X^{k}{\cal M}_{-p^{{e-h}}}(b(X))
 =  X^k\sum\limits_{j=0}^{n-1}b_j^{p^{{e-h}}}\! X^{-j}\\
& \equiv \sum\limits_{j=0}^{k}b_j^{p^{{e-h}}}\! X^{k-j}
          ~+~\lambda^{p^{e-h}}\!
           \sum\limits_{j=k+1}^{n-1}b_j^{p^{{e-h}}}\! X^{n+k-j}
         \pmod{X^n-\lambda^{-p^{{e-h}}}}.
\end{align*}
In the right hand side,
we replace $k-j$ by $i$ in the first $\sum$,
and replacing $n+k-j$ by $i$ in the second $\sum$.
In the ring $R_{n,\lambda^{-p^{e-h}}}$, we get that
\begin{equation}\label{in R_n,lambda^}
  X^{k}{\cal M}_{-p^{{e-h}}}(b(X))
   =\sum\limits_{i=0}^{k}b_{k-i}^{p^{{e-h}}}\! X^{i}
   +\lambda^{p^{e-h}}\!\sum\limits_{i=k+1}^{n-1}b_{n+k-i}^{p^{{e-h}}}\! X^{i},
    ~~~  k=0,1,\cdots  ,n-1.
\end{equation}
Note that $(f^{p^{{e-h}}}) ^{p^{h}}=f$ for any $f\in{\Bbb F}_q$.
From Eqn \eqref{in R_n,lambda} and Eqn \eqref{in R_n,lambda^} we obtain that
$$ a(X)b(X)=0~~\mbox{in $R_{n,\lambda}$}
  ~\iff~\big\langle{\cal M}_{-p^{{e-h}}}(b(X))\big\rangle
   \subseteq a(X)^{\perp h}~~ \mbox{in $R_{n,\lambda^{-p^{e-h}}}$.}  $$
Thus
$${\cal M}_{-p^{{e-h}}}({\rm Ann}(C))\subseteq C^{\perp h}. $$
The inclusion has to be an equality
because $\dim_{{\Bbb F}_q}(C^{\perp})=\dim_{{\Bbb F}_q}({\rm Ann}(C))$.
\end{proof}

Combining the above lemma with
Lemma \ref{Ann=bar} and Theorem \ref{M_s phi},
we get the following theorem and corollary at once.

\begin{Theorem}\label{C_phi bot}
Let $C_\varphi\le R_{n,\lambda}$ be a $\lambda$-constacyclic
code with check polynomial $f_\varphi(X)$, where
$\varphi:(1+r{\Bbb Z}_{n'r})/\mu_q\to[0,p^{\nu_p(n)}]$ is a $q$-coset function.
Then
$$
 C_\varphi^{\bot h}={\cal M}_{-p^{e-h}}(C_{\bar\varphi})
  =C_{-p^{e-h}\bar\varphi}
$$
which is a $\lambda^{-p^{e-h}}$-constacyclic code.
\end{Theorem}

\begin{Corollary}\label{Euc-Herm}
{\bf(i)}
  The dual code $C_\varphi^{\bot 0}\!=\!C_{-\bar\varphi}$,
  which is a $\lambda^{-1}$-constacyclic code.

 {\bf(ii)}
  The Hermitian dual code
  $C_\varphi^{\bot \frac{e}{2}}\!=\!C_{-p^{\frac{e}{2}}\bar\varphi}$,
  which is a $\lambda^{-p^\frac{e}{2}}$-constacyclic code.
\end{Corollary}

In the semisimple case (i.e. $\nu_p(n)=0$),
the conclusion (i) of the corollary was proved in \cite{Bl13}.
On the other hand, it has been shown in \cite{D, JLS}
that the Hermitian dual code of a $\lambda$-constacyclic code is
a $\lambda^{-p^\frac{e}{2}}$-constacyclic code.

\section{\!\!\!\!Isometrically Galois self-dual constacyclic codes}

\begin{Definition}\label{def h-inner product}\rm
Let $C_\varphi\leq R_{n,\lambda}$ be a $\lambda$-constacyclic code
with check polynomial $f_\varphi(X)$, where
$\varphi:(1+r{\Bbb Z}_{n'r})/\mu_q\to[0,p^{\nu_p(n)}]$ is a $q$-coset function.
\begin{itemize}
\item[\bf(i)]
If  $C_\varphi=C_\varphi^{\bot h}$,
then we say that $C_\varphi$ is a {\em Galois self-dual}
(more explicitly, {\em $p^h$-self-dual}) $\lambda$-constacyclic code.
\item[\bf(ii)]
If there is an integer $s$ with $\gcd(s,n'r)=1$ and $s\equiv 1\!\pmod{r}$
such that
$${\cal M}_{-p^{e-h}s}(C_\varphi)=C_\varphi^{\bot h},$$
then we say that $C_\varphi$ is {\em isometrically Galois self-dual}
(more explicitly, {\em isometrically $p^h$-self-dual}).
\end{itemize}
\end{Definition}

The $p^0$-self-dual constacyclic codes
are just the usual self-dual constacyclic codes,
which were studied by many researchers, e.g., \cite{Bl08, DL, DP}.
The isometrically $p^0$-self-dual constacyclic codes are
the so-called iso-dual constacyclic codes studied in \cite{Bl13}.
And, the $p^{\frac{e}{2}}$-self-dual constacyclic codes (if $e$ is even) are
the usual Hermitian self-dual constacyclic codes considered in \cite{Bl13}.

Recall that a linear code is said to be {\em formal self-dual}
if the code and its dual code have one and the same weight distribution,
cf. \cite[p.307]{HP}.
Any isometrically Galois self-dual constacyclic code
is obviously formal self-dual

In this section we'll exhibit a necessary and sufficient condition for
the existence of isometrically Galois self-dual constacyclic codes.

\begin{Lemma}\label{s phi}
Let $s$ be an integer with $\gcd(s,n'r)=1$ and $s\equiv 1\!\pmod r$.
Let $\varphi:(1+r{\Bbb Z}_{n'r})/\mu_q \rightarrow [0,p^{\nu_p(n)}]$ 
be a $q$-coset function. Then the following four statements 
are equivalent to each other:
\begin{itemize}
\item[\bf(i)]
${\cal M}_{-p^{e-h}s}(C_\varphi)=C_\varphi^{\bot h}$.
\item[\bf(ii)]
$s\varphi=\bar{\varphi}$.
\item[\bf(iii)]
$\varphi(Q)+\varphi(sQ)=p^{\nu_p(n)}$,~
  $\forall$ $Q\in(1+r{\Bbb Z}_{n'r})/\mu_q$.
\item[\bf(iv)]
For any $s$-orbit $Q,sQ,\cdots, s^{\ell-1}Q$ of length $\ell$
on $(1+r{\Bbb Z}_{n'r})/\mu_q$ (see Remark~\ref{q-coset}~(ii)),
one of the following two holds:
 \begin{itemize}
  \item[\bf(iv.a)]
  $\ell$ is even,
  $\varphi(Q)=\varphi(s^2Q)=\cdots=\varphi(s^{\ell-2}Q)$,
  $\varphi(sQ)=\varphi(s^3Q)=\cdots=\varphi(s^{\ell-1}Q)$,
  and $\varphi(Q)+\varphi(sQ)=p^{\nu_p(n)}$.
 \item[\bf(iv.b)]
  $\ell$ is odd,
  $\varphi(Q)=\varphi(sQ)=\varphi(s^2Q)=\cdots=\varphi(s^{\ell-1}Q)
   =\frac{1}{2}p^{\nu_p(n)}$.
  \end{itemize}
\end{itemize}
\end{Lemma}

\begin{proof}
{\bf(i) $\iff$ (ii).}~ By Theorem \ref{C_phi bot},
$C_\varphi^{\bot h}=C_{-p^{e-h}\bar\varphi}$.
On the other hand, by Theorem~\ref{M_s phi},
${\cal M}_{-p^{e-h}s}(C_\varphi)=C_{-p^{e-h}s\varphi}$. So
${\cal M}_{-p^{e-h}s}(C_\varphi)=C_\varphi^{\bot h}$
if and only if
$C_{-p^{e-h}s\varphi}=C_{-p^{e-h}\bar\varphi}$,
if and only if
$-p^{e-h}s\varphi= -p^{e-h}\bar\varphi$.
Note that $-p^{e-h}\in{\Bbb Z}_{n'r}^*$. So
$-p^{e-h}s\varphi= -p^{e-h}\bar\varphi$ if and only if
$s\varphi=\bar\varphi$.

{\bf(ii) $\iff$ (iii).}~ 
Let $s$ be an integer such that $s^{-1}s\equiv 1\!\pmod{n'r}$,
let $Q\in(1+r{\Bbb Z}_{n'r})/\mu_q$.
(ii) can be rewritten as $\varphi=s^{-1}\bar\varphi$. If it holds, then
$\varphi(Q)=s^{-1}\bar\varphi(Q)
 =\bar\varphi(sQ)=p^{\nu_q(n)}-\varphi(sQ)$,
see Definition \ref{def q-coset function};
i.e., (iii) holds. If (iii) holds, then
$\varphi(Q)=p^{\nu_q(n)}-\varphi(sQ)
 =\bar\varphi(sQ)=s^{-1}\bar\varphi(Q)$,
i.e., (ii) holds.

{\bf(iii) $\implies$ (iv).}~
By (iii), $\varphi(Q)+\varphi(sQ)=p^{\nu_p(n)}
 =\varphi(sQ)+\varphi(s^2Q)$.
We get $\varphi(Q)=\varphi(s^2Q)$.
Similarly, $\varphi(sQ)+\varphi(s^2Q) =\varphi(s^2Q)+\varphi(s^3Q)$.
we get $\varphi(sQ)=\varphi(s^3Q)$.
Iterating in this way, we see that
\begin{equation}\label{s^iQ}
 \varphi(s^iQ)=\begin{cases}\varphi(Q), & \mbox{$i$ is even;}\\
 \varphi(sQ), & \mbox{$i$ is odd.}\end{cases}
\end{equation}

If $\ell$ is even, then, noting that $\ell-2$ is even while $\ell-1$ is odd, we get:
$$
\varphi(Q)=\varphi(s^2Q)=\cdots=\varphi(s^{\ell-2}),~~
\varphi(sQ)=\varphi(s^3Q)=\cdots=\varphi(s^{\ell-1}Q).
$$

Otherwise, $\ell$ is odd. Since $s^\ell Q=Q$
(see Remark~\ref{q-coset}~(ii)), by Eqn \eqref{s^iQ} we obtain that
$$
\varphi(sQ)=\varphi(s^{\ell}Q)=\varphi(Q).
$$
Combining it with (iii),
we further get that $\varphi(Q)=\varphi(sQ)=\frac{1}{2}p^{\nu_p(n)}$.
By Eqn~\eqref{s^iQ}, $\varphi(s^iQ)=\frac{1}{2}p^{\nu_p(n)}$ for any
integer $i$.

{\bf(iv) $\implies$ (iii).}~
Obviously, (iv) implies that, for any $s$-orbit $Q,sQ,\cdots, s^{\ell-1}Q$
of length $\ell$ on $(1+r{\Bbb Z}_{n'r})/\mu_q$, we have
$$
\varphi(s^iQ)+\varphi(s^{i+1}Q)=p^{\nu_p(n)},~~~~ i=0,1,\cdots.
$$
Thus, (iii) holds.
\end{proof}

A characterization of  isometrically $p^h$-self-dual constacyclic codes 
by $q$-coset functions is obviously obtained as follows.

\begin{Corollary}
$C_\varphi$ is isometrically $p^h$-self-dual if and only if there is an integer~$s$
with $\gcd(s,n'r)=1$ and $s\equiv 1\!\pmod{r}$ such that
$s\varphi=\bar\varphi$.
\end{Corollary}

\begin{Lemma}\label{=bar phi}
Let $s$ be an integer with $\gcd(s,n'r)=1$ and $s\equiv 1\!\pmod r$.
The following two statements are equivalent to each other:
\begin{itemize}
\item[\bf(i)]
There exists a $q$-coset function
$\varphi: (1+r{\Bbb Z}_{n'r})/\mu_q\rightarrow [0,p^{\nu_p(n)}]$
such that $s\varphi=\bar\varphi$.
\item[\bf(ii)]
 One of the following two conditions holds:
\begin{itemize}
\item[\bf(ii.1)]
$p=2$ and $\nu_2(n)\ge 1$.
\item[\bf(ii.2)]
The length of any $s$-orbit on $(1+r{\Bbb Z}_{n'r})/\mu_q$ is even.
\end{itemize}
\end{itemize}
\end{Lemma}

\begin{proof}
(i)$\implies$(ii). Assume that (i) holds and (ii.2) is not satisfied,
i.e., $s\varphi=\bar\varphi$ but there is at least one
$s$-orbit $Q, sQ,\cdots,s^{\ell-1}Q$ on $(1+r{\Bbb Z}_{n'r})/\mu_q$
whose length $\ell$ is odd. Then, by Lemma \ref{s phi},
$\varphi(Q)=\frac{1}{2}p^{\nu_p(n)}$. So it has to be the case that
the prime $p=2$ and $\nu_2(n)\ge 1$; i.e., (ii.1) is satisfied.

(ii)$\implies$(i).
First assume that the condition (ii.1) holds. We take a $q$-coset function
$\varphi: (1+r{\Bbb Z}_{n'r})/\mu_q \to [0,2^{\nu_2(n)}]$ as follows:
\begin{equation}\label{1/2 p}
 \varphi(Q)=\frac{1}{2}\cdot 2^{\nu_2(n)},~~~~ \forall~
  Q\in(1+r{\Bbb Z}_{n'r})/\mu_q.
\end{equation}
Then Lemma \ref{s phi} (iii) is satisfied obviously,
so $s\varphi=\bar\varphi$.

Next assume that the condition (ii.2) holds. We take an integer $d$
such that $0\le d<p^{\nu_p(n)}$,
and define a $q$-coset function
$\varphi:(1+r{\Bbb Z}_{nr})/\mu_q\to[0,p^{\nu_p(n)}]$ as follows:
for each $s$-orbit $Q,sQ,\cdots,s^{\ell-1}Q$ 
of length $\ell$ on $(1+r{\Bbb Z}_{nr})/\mu_q$,
since $\ell$ is even, we can set
\begin{equation}\label{even length}
 \varphi(s^iQ)=
 \begin{cases}d, &\mbox{$i$ is even};\\
 p^{\nu_p(n)}-d, &\mbox{$i$ is odd}.\end{cases}
\end{equation}
Then the condition (iv.a) of Lemma \ref{s phi} holds for all
$s$-orbits on $(1+r{\Bbb Z}_{nr})/\mu_q$. Thus, by Lemma \ref{s phi},
$s\varphi=\bar\varphi$.
\end{proof}

We state some facts for the semisimple case which come 
from \cite{CDFL}. Note that a duadic $\lambda'$-constacyclic 
code over ${\Bbb F}_q$ of length $n'$ is corresponding 
to a partition $(1+r{\Bbb Z}_{n'r})/\mu_q={\cal X}\cup{\cal X}'$ 
and an $s\in{\Bbb Z}^*_{n'r}\cap(1+r{\Bbb Z}_{n'r})$ such that
$s{\cal X}={\cal X}'$.

\begin{Lemma}\label{duadic} 
The following three statements are equivalent to each other:
\begin{itemize}
\item[\bf(i)]
There is an integer $s$ with $\gcd(s,n'r)=1$ and $s\equiv 1\!\pmod r$
such that the length of any $s$-orbit on $(1+r{\Bbb Z}_{n'r})/\mu_q$ is even.
\item[\bf(ii)]
The duadic $\lambda'$-constacyclic codes over ${\Bbb F}_q$ of length $n'$ exist.
\item[\bf(iii)]
$q$ is odd and one of the following two conditions holds:
\begin{itemize}
\item[\bf(iii.1)]
 $\nu_2(n')\ge 1$ and $\nu_2(q-1)>\nu_2(r)\ge 1$;
\item[\bf(iii.2)]
$\nu_2(r)=1$ and $\min\{\nu_2(q+1),\nu_2(n')\}\ge 2$.
\end{itemize}
\end{itemize}
\end{Lemma}

\begin{proof}
From \cite[Lemma 6]{CDFL}
we can get the equivalence of (i) and (ii).
By \cite[Corollary 14]{CDFL}, (ii) is equivalent to (iii).
\end{proof}

\begin{Theorem}\label{iso h-self dual}
 The isometrically $p^h$-self-dual $\lambda$-constacyclic codes
over ${\Bbb F}_{q}$ of length $n$ exist if and only if
one of the following three conditions holds:
\begin{itemize}
\item[\bf(i)]
$p=2$ and $\nu_2(n)\ge 1$.
\item[\bf(ii)]
$p$ is odd, $\nu_2(n')\ge 1$ and $\nu_2(q-1)>\nu_2(r)\ge 1$.
\item[\bf(iii)]
$p$ is odd, $\nu_2(r)=1$ and $\min\{\nu_2(q+1),\nu_2(n')\}\ge 2$.
\end{itemize}
\end{Theorem}

\begin{proof}
First we prove the necessity. Assume that
$C_\varphi\le R_{n,\lambda}$ is an isometrically $p^h$-self-dual
$\lambda$-constacyclic code. By Lemma \ref{s phi},
there is an integer~$s$ with $\gcd(s,n'r)=1$ and $s\equiv 1\!\pmod{r}$
such that $s\varphi=\bar\varphi$. By Lemma \ref{=bar phi},
either (i) of the theorem holds, or
the length of any $s$-orbit on $(1+r{\Bbb Z}_{n'r})/\mu_q$ is even,
hence, by Lemma \ref{duadic} (iii), one of the (ii) and (iii) of the theorem holds.

\smallskip
Next we prove the sufficiency.
If one of the conditions (ii) and (iii) holds,
then, by Lemma~\ref{duadic},
there is an integer $s$ with $\gcd(s,n'r)=1$ and $s\equiv 1\!\pmod r$
such that the length of any $s$-orbit on $(1+r{\Bbb Z}_{n'r})/\mu_q$ is even.
Thus, the sufficiency is deduced from 
Lemma \ref{=bar phi} and Lemma \ref{s phi} at once.
\end{proof}

\begin{Corollary}\label{one for all}
The following three statements are equivalent:
\begin{itemize}
\item[\bf(i)]
There is an $h\in[0,e]$ such that
the isometrically $p^h$-self-dual $\lambda$-constacyclic codes
of length $n$ over ${\Bbb F}_q$ exist.
\item[\bf(ii)]
For any $h\in[0,e]$, the isometrically $p^h$-self-dual $\lambda$-constacyclic
codes of length $n$ over ${\Bbb F}_q$ exist.
\item[\bf(iii)]
either $p=2$ and $\nu_2(n)\ge 1$, or
the duadic $\lambda'$-constacyclic codes over ${\Bbb F}_q$ of length $n'$ exist.
\end{itemize}
\end{Corollary}

\begin{proof}
The necessary and sufficient condition for the existence
of isometrically $p^h$-self-dual constacyclic codes 
stated in Theorem \ref{iso h-self dual}
is independent of the choice of $h\in[0,e]$;
hence the equivalence of (i) and (ii) is obtained.
And, by Lemma~\ref{duadic}, the statement (iii) is equivalent 
to the existence condition stated in Theorem~\ref{iso h-self dual}.
\end{proof}

\section{Galois self-dual constacyclic codes}

In this section we show a necessary and sufficient condition for the existence of
Galois self-dual constacyclic codes. 
We begin with a characterization of the Galois self-dual 
constacyclic codes by $q$-coset functions.

\begin{Theorem}\label{theorem s phi dual}
Let $h\in[0,e]$, and 
$\varphi:(1+r{\Bbb Z}_{n'r})/\mu_q\rightarrow [0,p^{\nu_p(n)}]$ 
be a $q$-coset function.
Then the following two statements are equivalent to each other.
\begin{itemize}
\item[\bf(i)]
$C_\varphi=C_\varphi^{\bot h}$~ (i.e., $C_\varphi$ is $p^h$-self-dual).
\item[\bf(ii)]
$r|\gcd(p^{h}+1,p^e-1)$ (i.e., $-p^{h}$ and $p^e\equiv 1\!\pmod{r}$)
and $-p^{h}\varphi=\bar\varphi$.
\end{itemize}
\end{Theorem}

\begin{proof}
By Theorem~\ref{C_phi bot}, 
$C_{\varphi}^{\bot h}=C_{-p^{e-h}\bar\varphi}$
which is a $\lambda^{-p^{e-h}}$-constacyclic code. We get that
 $C_\varphi=C_\varphi^{\bot h}$ if and only if
$\lambda^{-p^{e-h}}=\lambda$ and $C_\varphi=C_{-p^{e-h}\bar\varphi}$.
That is,  $C_\varphi=C_\varphi^{\bot h}$ if and only if
$-p^{e-h}\equiv 1\!\pmod{r}$ and $\varphi=-p^{e-h}\bar\varphi$.
Since $r|(q-1)$ where $q=p^e$, we have $p^e\equiv 1\!\pmod r$. Hence 
$$(-p^h)(-p^{e-h})=p^{e}\equiv 1\!\pmod r.$$
We see that $-p^{e-h}\equiv 1\!\!\pmod r$ 
if and only if  $-p^h\equiv 1\!\!\pmod r$.
Multiplying $-p^h$ to the both sides of the equality
$\varphi=-p^{e-h}\bar\varphi$, we get
$-p^h\varphi=(-p^h)(-p^{e-h})\bar\varphi$.
However, by Corollary \ref{s_1=s_2},
$(-p^h)(-p^{e-h})\bar\varphi=\bar\varphi$.
In conclusion, $C_{\varphi}=C_{\varphi}^{\bot h}$ if and only if
$-p^h\equiv 1\!\pmod r$ and $-p^h\varphi=\bar\varphi$.
\end{proof}

We need a number-theoretic result.

\begin{Lemma}\label{k^d}
Let $k$ be an odd integer, and $d$ be a positive integer.
\begin{itemize}
\item[\bf(i)]
If $k=1+2^vu$ with $v\ge 2$ and $2\nmid u$
(equivalently, $k\equiv 1\!\pmod 4$), then
$$\nu_2(k^d-1)=v+\nu_2(d), ~~~~~ \nu_2(k^d+1)=1.$$
\item[\bf(ii)]
If $k=-1+2^vu$ with $v\ge 2$ and $2\nmid u$
(equivalently, $k\equiv-1\!\pmod 4$), then
 \begin{itemize}
 \item[\bf(ii.1)]
  if $\nu_2(d)=0$ (i.e., $d$ is odd) then
   $$\nu_2(k^d-1)=1,~~~~~ \nu_2(k^d+1)=v;$$
 \item[\bf(ii.2)]
  if $\nu_2(d)\ge 1$ (i.e., $d$ is even) then
   $$\nu_2(k^d-1)=v+\nu_2(d), ~~~~~\nu_2(k^d+1)=1.$$
 \end{itemize}
\end{itemize}
\end{Lemma}

\begin{proof}
(i).~
If $d$ is a prime integer;
by the Newton's binomial formula it is easy to
check that
$$
 k^d=\begin{cases}1+2^{v+1}u', & d=2;\\ 1+2^vu', & d\ne 2;\end{cases}
  ~~~~~ \mbox{with}~ 2\nmid u'.
$$
For any positive integer $d$,
decomposing $d$ into a product of primes, we can get
$$ k^d=1+2^{v+\nu_2(d)}u',~~~~~ \mbox{with}~ 2\nmid u'.$$
Then it is obvious that
$\nu_2(k^d-1)=v+\nu_2(d)$ and $\nu_2(k^d+1)=1$.

(ii).~
Similarly to the above, if $d$ is a prime integer then
$$k^d=\begin{cases}1+2^{v+1}u', & d=2;\\
   -1+2^vu', & d\ne 2;\end{cases}~~~~~ \mbox{with}~ 2\nmid u'.$$
For any positive integer $d$, similarly to the above argument again,
$$k^d=\begin{cases}1+2^{v+\nu_2(d)}u', & \nu_2(d)\ge 1;\\
   -1+2^vu', & \nu_2(d)=0;\end{cases}~~~~~ \mbox{with}~ 2\nmid u'.$$
Then both (ii.1) and (ii.2) are easily derived.
\end{proof}

We return to our notations on constacyclic codes.

\begin{Lemma}\label{duadic given by}
Assume that $-p^{h}\equiv 1\!\pmod r$.
The length of any $(-p^{h})$-orbit on
$(1+r{\Bbb Z}_{n'r})/\mu_q$ is even
if and only if both $n'$ and $r$ are even (hence $p$ is odd)
and one of the following three conditions holds:
\begin{itemize}
\item[\bf(i)]
$p=1+2^vu$ with $v\ge 2$ and $2\nmid u$
(equivalently, $p\equiv 1\!\pmod 4$).
\item[\bf(ii)]
$p=-1+2^vu$ with $v\ge 2$ and $2\nmid u$
(equivalently, $p\equiv -1\!\pmod 4$),
both $e$ and $h$ are even.
\item[\bf(iii)]
$p=-1+2^vu$ with $v\ge 2$ and $2\nmid u$
(equivalently, $p\equiv -1\!\pmod 4$),
at least one of $e$ and $h$ is odd, and $\nu_2(n'r)>v$.
\end{itemize}
\end{Lemma}

\begin{proof} Note that $q=p^e$.
By \cite[Lemma 6]{CDFL}, the length of any $(-p^{h})$-orbits
on $(1+r{\Bbb Z}_{n'r})/\mu_q$ is even if and only if
the duadic $\lambda$-constacyclic codes over ${\Bbb F}_{p^e}$ of length $n'$
given by the multiplier $\mu_{-p^{h}}$ exist.
Further, by \cite[Corollary 19]{CDFL}, the latter statement holds if and only if
both $n'$ and $r$ are even (hence $q=p^e$ is odd)
and one of the following four conditions holds
(we adopt a convention that $\nu_2(0)=-\infty$ hence $|\nu_2(0)|=\infty$):
\begin{itemize}
\item[(c1)]
$\nu_2(p^e-1)>\nu_2(-p^{h}-1)$ and $\nu_2(n'r)>\nu_2(-p^{h}-1)$;
\item[(c2)]
$\nu_2(p^e-1)=1$, $\nu_2(-p^{h}-1)>1$,
 $\nu_{2}(p^e+1)+1>\nu_2(-p^{h}-1)$
 and $\nu_{2}(n'r)>\nu_2(-p^{h}-1)$;
\item[(c3)]
$\nu_2(p^e-1)=\nu_2(-p^{h}-1)=1$,~
 $|\nu_2(-p^{h}+1)|>\nu_{2}(p^e+1)$ and $\nu_2(n'r)>\nu_{2}(p^e+1)$;
\item[(c4)]
$\nu_2(p^e-1)=\nu_2(-p^{h}-1)=1$,~
 $|\nu_2(-p^{h}+1)|<\nu_{2}(p^e+1)$
 and $|\nu_2(-p^{h}+1)|<\nu_{2}(n'r)$.
\end{itemize}
It remains to show that one of the four conditions holds
if and only if one of (i), (ii) and (iii) of the lemma holds.
We discuss it in two cases.

{\it Case 1}:  $p=1+2^vu$ with $v\ge 2$ and $2\nmid u$.
At this case the condition (i) of the lemma holds.
On the other hand, by Lemma \ref{k^d},
$$
  \nu_2(p^e-1)=v+\nu_2(e)~~~\mbox{and}~~~
   \nu_2(-p^{h}-1)=\nu_2(p^{h}+1)=1.
$$
Note that $\nu_2(n'r)\ge 2$ (since both $n'$ and $r$ are even),
the condition (c1) holds.

{\it Case 2}:  $p=-1+2^vu$ with $v\ge 2$ and $2\nmid u$.
There are four subcases.

{\it Subcase 2.1}: both $e$ and $h$ are even. By Lemma \ref{k^d},
$$
 \nu_2(p^e-1)=v+\nu_2(e), ~~~\mbox{and}~~~
 \nu_2(-p^{h}-1)=\nu_2(p^{h}+1)=1.
$$
The condition (c1) holds, and the condition (ii) of the lemma holds.

{\it Subcase 2.2}: $e$ is even, $h$ is odd.
As we have seen, $\nu_2(p^e-1)=v+\nu_2(e)$. None of (c2), (c3) and (c4) holds.
Further, since $h$ is odd, by Lemma \ref{k^d},
$\nu_2(-p^{h}-1)=\nu_2(p^{h}+1)=v$.
So, (c1) holds if and only if $\nu_2(n'r)>\nu_2(-p^{h}-1)=v$;
and if it is, (iii) of the lemma also holds.

{\it Subcase 2.3}: $e$ is odd, $h$ is even.
Then $\nu_2(p^e-1)=1$ so that (c1) cannot hold.
And, since $\nu_2(-p^{h}-1)=\nu_2(p^{h}+1)=1$,
(c2) does not hold. Further,
$$
\nu_2(-p^{h}+1)=\nu_2(p^{h}-1)=v+\nu_2(h)  > v = \nu_2(p^e+1),
$$
which implies that (c4) is not satisfied.
The condition (c3) is satisfied provided  $\nu_2(n'r)>\nu_2(p^e+1)=v$,
which is also required by (iii) of the lemma.

{\it Subcase 2.4}: $e$ is odd, and $h$ is odd. Then
$$
 \nu_2(p^e-1)=1, ~~ \nu_2(p^e+1)=v, ~~
  \nu_2(p^{h}+1)=v, ~~ \nu_2(p^{h}-1)=1.
$$
None of the conditions (c1), (c3) and (c4) holds.
Note that one more requirement ``$\nu_2(n'r)>\nu_2(p^{h}+1)=v$''
makes (c2) held, and it also makes (iii) of the lemma held.
\end{proof}

\begin{Theorem}\label{h-self dual}
The $p^h$-self-dual $\lambda$-constacyclic codes over ${\Bbb F}_{q}$
of length $n$ exist if and only if $r\,\big|\,\gcd(p^{h}+1,\,p^e-1)$
and one of the following holds:
\begin{itemize}
\item[\bf(i)]
$p=2$ and $\nu_2(n)\ge 1$.
\item[\bf(ii)]
$p=1+2^vu$ with $v\ge 2$ and $2\nmid u$
(equivalently, $p\equiv 1\!\pmod 4$),
both $n'$ and $r$ are even.
\item[\bf(iii)]
$p=-1+2^vu$ with $v\ge 2$ and $2\nmid u$
(equivalently, $p\equiv -1\!\pmod 4$),
all of $n'$, $r$, $e$ and $h$ are even.
\item[\bf(iv)]
$p=-1+2^vu$ with $v\ge 2$ and $2\nmid u$
(equivalently, $p\equiv -1\!\pmod 4$),
both $n'$ and $r$ are even, but
at least one of $e$ and $h$ is odd, and $\nu_2(n'r)>v$.
\end{itemize}
\end{Theorem}

\begin{proof}
By Theorem \ref{theorem s phi dual},
the $p^h$-self-dual $\lambda$-constacyclic codes
over ${\Bbb F}_{q}$ of length $n$ exist if and only if $r\,|\,(p^{h}+1)$
and there is a $q$-coset function
$\varphi:(1+r{\Bbb Z}_{n'r})/\mu_q\to[0,p^{\nu_p(n)}]$
such that $-p^h\varphi=\bar\varphi$.
By Lemma \ref{=bar phi}, $-p^h\varphi=\bar\varphi$
if and only if either (i) of the theorem holds
or the length of any $(-p^h)$-orbit on $(1+r{\Bbb Z}_{n'r})/\mu_q$ is even.
Further, the length of any $(-p^{h})$-orbit
on $(1+r{\Bbb Z}_{n'r})/\mu_q$ is even if and only if
one of the conditions (i), (ii) and (iii) in Lemma \ref{duadic given by}
holds, they are restated in the theorem relabeled by (ii), (iii) and (iv).
\end{proof}

\begin{Corollary}\label{self-dual}
Self-dual $\lambda$-constacyclic codes
over ${\Bbb F}_{q}$ of length $n$ exist if and only if
one of the following holds:
\begin{itemize}
\item[\bf(i)]
$p=2$, $\lambda=1$  and $\nu_2(n)\ge 1$.
\item[\bf(ii)]
$p^e\equiv 1\!\pmod 4$, $\lambda=-1$, $n'$ is even.
\item[\bf(iii)]
$p^e\equiv -1\!\pmod 4$, $\lambda=-1$,
and $\nu_2(n')+1>\nu_2(p^e+1)$.
\end{itemize}
\end{Corollary}

\begin{proof}
Take $h=0$ in Theorem \ref{h-self dual}.
The condition that $r\,|\,\gcd(p^h+1,p^e-1)$ implies that $r\,|\, 2$.
Then Theorem \ref{h-self dual} (i) is reduced to $p=2$, $\nu_2(n)\ge 1$,
hence $r=1$.
And Theorem \ref{h-self dual} (ii) and (iii) are reduced to
$p^e\equiv 1\!\pmod 4$, $r=2$ and $n'$ is even.
Finally, Theorem \ref{h-self dual} (iv) is reduced to
$p^e\equiv -1\!\pmod 4$, $r=2$ and $\nu_2(n')+1>\nu_2(p^e+1)$.
\end{proof}

We remark that (i) of Corollary \ref{self-dual} is 
the cyclic (but not semisimple) case.
In \cite{W} there is a similar result for any group codes.
On the other hand, if it is restricted to the semisimple case,
then Corollary \ref{self-dual} (i) is not allowed,
and $n'=n$ (recall that $n=p^{\nu_p(n)}n'$).
So \cite[Theorem 3]{Bl08} or \cite[Corollary 21]{CDFL} are
obtained as a consequence of Corollary \ref{self-dual}.

\begin{Corollary}\label{Hermitian self-dual}
Hermitian self-dual $\lambda$-constacyclic codes
over ${\Bbb F}_{q}$ of length $n$ exist if and only if
$e$ is even, $r\,\big|\,\gcd(p^{\frac{e}{2}}+1,p^e-1)$
and one of the following holds:
\begin{itemize}
\item[\bf(i)]
$p=2$ and $\nu_2(n)\ge 1$.
\item[\bf(ii)]
$p^{\frac{e}{2}}\equiv 1\!\pmod 4$,
both $n'$ and $r$ are even.
\item[\bf(iii)]
$p^{\frac{e}{2}}\equiv -1\!\pmod 4$,
both $n'$ and $r$ are even, $\nu_2(n'r)>\nu_2(p^{\frac{e}{2}}+1)$.
\end{itemize}
\end{Corollary}

\begin{proof}
In Theorem \ref{h-self dual}, let $e$ be even and $h=\frac{e}{2}$.
So $r\,\big|\,\gcd(p^{\frac{e}{2}}+1,p^e-1)$.
Then (i) of the corollary is just the (i) of Theorem \ref{h-self dual}.
By Lemma \ref{k^d}, Theorem~\ref{h-self dual} (ii) and (iii) are reduced to
the (ii) of the corollary. Finally, when $p=-1+2^v u$ with $v\ge 2$ and $2\nmid u$,
$\frac{e}{2}$ is odd if and only if $p^{\frac{e}{2}}\equiv -1\!\pmod 4$;
and at that case, $v=\nu_{2}(p^{\frac{e}{2}}+1)$; see Lemma \ref{k^d}.
So Theorem \ref{h-self dual} (iv) is reduced to (iii) of the corollary.
\end{proof}

\section{Examples}

The first example is constructed in the way of Eqn \eqref{1/2 p} to
illustrates the repeated-root case where $p=2$.

\begin{Example}\rm
Let $p=2$, $e=2$, i.e., $q=4$, and let  $\theta\in{\Bbb F}_4$
be a primitive third root of unity, i.e.,
${\Bbb F}_4=\{0,1,\theta,\theta^2\}$ and $\theta^2+\theta+1=0$.
Let $n=2$, hence $\nu_2(n)=1$ and $n'=1$.

(i)~
Take $\lambda=\theta^2$ (so $r=3$).
Then $r\,\big|\,\gcd(2^1+1,2^2-1)$ and Corollary~\ref{Hermitian self-dual}~(i)
holds, so an Hermitian self-dual $\theta^2$-constacyclic code exists
(but the self-dual $\theta^2$-constacyclic codes do not exist).
In fact, $X^2-\theta^2=(X-\theta)^2$, and $1+r{\Bbb Z}_{n'r}=\{1\}$.
The $q$-coset function $\varphi(1)=1$ (then $\bar\varphi=\varphi$) 
is corresponding to the $\theta^2$-constacyclic code 
$C_\varphi\le R_{2,\theta^2}$ generated by $X+\theta$, i.e.,
$$C_\varphi=\langle X+\theta\rangle
 =\big\{(0,0),(\theta, 1), (\theta^2,\theta), (1,\theta^2)\big\},$$
which is Hermitian self-dual since $-3\varphi=\bar\varphi$. 
Also, a direct computation is as follows:
$$
 \big\langle (\theta,1),(\theta,1)\big\rangle_1
  =\theta\cdot\theta^2+1\cdot 1^2=1+1=0.
$$

(ii)~ However, if take $\lambda=1$ (i.e., $r=1$),
then a self-dual cyclic code exists
(which is also an Hermitian self-dual cyclic code) as follows:
$X^2+1=(X+1)^2$, $1+r{\Bbb Z}_{n'r}=\{1\}$, take $q$-coset function
$\psi(1)=1$, then $C_\psi\le R_{2,1}$ is a self-dual cyclic code.
\end{Example}

The next example is also the repeated-root case, but it is
constructed in the way of Eqn \eqref{even length}.

\begin{Example}\rm
Let $p=3$, $e=4$ hence $q=3^4$.
Then ${\Bbb F}_q$ contains a primitive $16$-th root of unity.
Take $\lambda=\theta^{12}$ and $n=3\cdot 4$.
Hence $r=4$, $\lambda'=\theta^{4}$, $\nu_3(n)=1$, $n'=4$,
$[0,3^{\nu_3(n)}]=\{0,1,2,3\}$,
$1+r{\Bbb Z}_{n'r}=\{1,5,9,13\}$ whose elements
are all fixed by $\mu_q=\mu_{3^4}$, and
$$
  X^{12}-\lambda=(X^{4}-\theta^4)^3
   =(X-\theta)^3 (X-\theta^5)^3 (X-\theta^9)^3 (X-\theta^{13})^3 .
$$
Define a $q$-coset function
$\varphi :1+r{\Bbb Z}_{n'r}\rightarrow [0,3] $ by
$$\varphi(1)=1,~~~\varphi(5)=2,~~~\varphi(9)=1,~~~\varphi(13)=2. $$
Then
$$f_\varphi(X)=(X-\theta)(X-\theta^5)^2(X-\theta^9)(X-\theta^{13})^2. $$
We can consider the $\theta^{12}$-constacyclic code
$C_\varphi\le R_{12,\theta^{12}}$ with check polynomial $f_\varphi(X)$.
It is easy to check that $-3\varphi=\bar\varphi$.
We have the following conclusions.
\begin{itemize}
\item
By Theorem \ref{theorem s phi dual}, $C_\varphi$
is a $3^1$-self-dual $\theta^{12}$-constacyclic code.
\item
By Corollary \ref{self-dual} and Corollary \ref{Hermitian self-dual},
$C_\varphi$ is neither self-dual nor Hermitian self-dual,
because $r\ne 2$ and $r\nmid\gcd(3^2+1,3^4-1)$.
\item
By Lemma \ref{s phi}, 
for any $h\in[0,4]$, the $C_\varphi$ is an isometrically $3^h$-self-dual
$\theta^{12}$-constacyclic codes of length $12$ over ${\Bbb F}_{3^4}$.
For example, because (cf. Lemma~\ref{s phi})
$${\cal M}_{(-3^{\frac{4}{2}})(-3)}(C_\varphi)
={\cal M}_{-3^{\frac{4}{2}}}(C_{-3\varphi})
={\cal M}_{-3^{\frac{4}{2}}}(C_{\bar\varphi})
=C_\varphi^{\bot\frac{4}{2}},
$$
the code $C_\varphi$ is isometrically Hermitian self-dual.
\end{itemize}
\end{Example}

The following example shows that a constacyclic code
can be both self-dual and Hermitian self-dual.

\begin{Example}\rm
Let $p=3$, $q=9$ (i.e., $e=2$), $\lambda=-1$ (i.e., $r=2$) and $n=4$.
Then $\nu_p(n)=\nu_3(4)=0$ (i.e., it is the semisimple case: $n=n'=4$),
$n'r=4\cdot 2=8$, $1+r{\Bbb Z}_{n'r}=\{1,3,5,7\}$ 
on which $\mu_q=\mu_9$ is the identity permutation. 
Take a $q$-coset function $\varphi$ as follows:
$$
    \varphi(1)=\varphi(3)=0,~~~~
    \varphi(5)=\varphi(7)=1. $$
Then
$$ \bar\varphi(1)=\bar\varphi(3)=1,~~~~
    \bar\varphi(5)=\bar\varphi(7)=0. $$
It is easy to check that
$$
\varphi=-\bar\varphi,~~~~~ \varphi=-3\bar\varphi
$$
Thus, $C_\varphi\le R_{4,-1}$ is a $[4,2,3]$ negacyclic code over ${\Bbb F}_9$
which is both self-dual and Hermitian self-dual.
\end{Example}

In many cases the self-dual constacyclic codes and
Hermitian self-dual constacyclic codes both exist,
but there is no constacyclic code which is both self-dual and Hermitian self-dual.

\begin{Example}\label{ex q=25}\rm Let $p=5$, $e=2$, i.e., $q=5^2=25$.
Take $r=2$ (i.e., $\lambda=-1$), $n=26$.
Then $n'=n=26$ (i.e., $\nu_p(n)=0$), $n'r=52$,
$X^{26}+1$ has no multiple roots and
${\Bbb F}_{25}[X]/\langle X^{26}+1\rangle$ is semisimple. Consider
\begin{align*}
 1+r{\Bbb Z}_{n'r}=1+2{\Bbb Z}_{52}
 =\{&1,~\,3,~\,5,~\,7,\,~ 9,~11,13,15,17,19,21,23,25,\\
 & 27,29,31,33,35,37,39,41,43,45,47,49,51\}.
\end{align*}
The $q$-cosets are 
(where $Q_i=\{i,iq,iq^2,\cdots\}$ denotes the $q$-coset containing $i$):
\begin{align*}
 (1+r{\Bbb Z}_{n'r})/\mu_q=(1+2{\Bbb Z}_{52})/\mu_{25}
  =\big\{\,&Q_1,~Q_3,~ Q_5,~ Q_7,~ Q_9,~ Q_{11},~ Q_{13},\\
   & Q_{27},Q_{29},Q_{31},Q_{33},Q_{35},Q_{37},Q_{39}\,\big\}.
\end{align*}
where
$$\begin{array}{llll}
Q_1=\{1,25\}, & Q_3=\{3,23\}, & Q_5=\{5,21\},\\
Q_7=\{7,19\}, & Q_{9}=\{9,17\}, & Q_{11}=\{11,15\}, & Q_{13}=\{13\},\\
Q_{27}=\{27,51\}, & Q_{29}=\{29,49\}, & Q_{31}=\{31,47\},\\
Q_{33}=\{33,45\}, & Q_{35}=\{35,43\}, & Q_{37}=\{37,41\}, & Q_{39}=\{39\}\\
\end{array}$$
The orbits of $\mu_{-1}$ on $(1+2{\Bbb Z}_{52})/\mu_{25}$ are as follows:
$$
\{Q_1,Q_{27}\},  \{Q_3,Q_{29}\},  \{Q_5,Q_{31}\},
\{Q_7,Q_{33}\}, \{Q_{9},Q_{35}\}, \{Q_{11},Q_{37}\},
\{Q_{13},Q_{39}\}.
$$
The orbits of $\mu_{-5}$ on $(1+2{\Bbb Z}_{52})/\mu_{25}$ are as follows:
$$
\{Q_1,Q_{31}\},  \{Q_3,Q_{37}\},  \{Q_5,Q_{27}\},
\{Q_7,Q_{9}\}, \{Q_{11},Q_{29}\}, \{Q_{33},Q_{35}\},
\{Q_{13},Q_{39}\}.
$$
Correspondingly, 
we define two $q$-coset functions $\varphi_{-1}$, $\varphi_{-5}$ as follows:
$$
 \varphi_{-1}(Q_j)=
 \begin{cases} 0, & j=1,3,5,7,9,11,13;\\
                         1, & j=27,29,31,33,35,37,39.\end{cases}
$$
$$
 \varphi_{-5}(Q_j)=
 \begin{cases} 0, & j=1,3,5,7,11,13,33;\\
                         1, & j=9,27,29,31,35,37,39.\end{cases}
$$
Then the negacyclic code $C_{\varphi_{-1}}\le R_{26,-1}$ 
over ${\Bbb F}_{25}$ is $5^0$-self-dual (i.e., self-dual), 
but not $5^1$-self-dual (i.e., not Hermitian self-dual). 
On the other hand, the negacyclic code $C_{\varphi_{-5}}\le R_{26,-1}$ 
over ${\Bbb F}_{25}$ is $5^1$-self-dual (i.e., Hermitian self-dual), 
but not $5^0$-self-dual (i.e., not self-dual).
\end{Example}

\section*{Acknowledgements}
The research of the authors is supported
by NSFC with grant number 11271005.

\end{document}